\documentclass[letterpaper, conference, onecolumn, 11pt]{IEEEtran}
\usepackage{lics}
\usepackage[margin=0.9in]{geometry}

\usepackage{microtype}
\usepackage{graphicx}
\usepackage{latexsym}
\usepackage[utf8]{inputenc}
\usepackage{booktabs}
\usepackage{arydshln}
\usepackage{enumerate}
\usepackage{graphicx} 
\usepackage{amssymb,latexsym,times,amsmath,xspace}
\usepackage{enumerate,verbatim}
\usepackage{rotating}
\usepackage{environ}
\usepackage{multicol,esvect,relsize}
\usepackage{tikz,hyperref}
\usepackage{forest}

\usepackage{amsthm,amssymb}

\newcommand{\logicFont}[1]{\protect\ensuremath{\mathrm{#1}}\xspace}

\newcommand{\classFont}[1]{\protect\ensuremath{\mathsf{#1}}\xspace}

\newcommand{\FO}{\logicFont{FO}}

\newcommand{\PSPACE}{\classFont{PSPACE}}

\newcommand{\ER}{\exists \mathbb{R}}
\newcommand{\NP}{\classFont{NP}}

\newcommand{\PTIME}{\classFont{P}}

\newcommand{\A}{\mathfrak{A}}

\newcommand{\Q}{\mathbb{Q}}
\newcommand{\Qs}{\mathbb{Q}^{<}}

\newcommand{\T}{\mathcal{T}}

\newcommand{\calC}{\mathcal{C}}

\newcommand{\NN}{\mathcal{N}}

\newcommand{\G}{\mathcal{G}}

\newcommand{\dom}{\operatorname{dom}}

\newcommand{\RE}{\mathbb{R}}
\newcommand{\ERE}{\exists\mathbb{R}}
\newcommand{\EREs}{\exists\mathbb{R}^{<}}
\newcommand{\REs}{\mathbb{R}^{<}}

\newcommand{\ESO}{{\rm ESO}}

\newcommand{\ESOr}{\ESO_{\mathbb{R}}}

\newcommand{\EFO}{{\exists\FO}}

\newcommand{\vcol}{\mathrel{\mathop:}}
\newcommand{\dfn}{\vcol=}
\newcommand{\ddfn}{\vcol\vcol=}
\newcommand{\ddnf}{\ddfn}

\newcommand\bigexists{%
  \mathop{\lower0.75ex\hbox{\ensuremath{%
    \mathlarger{\mathlarger{\mathlarger{\mathlarger{\exists}}}}}}}%
  \limits}

\newtheorem{theorem}{Theorem}
\newtheorem{definition}[theorem]{Definition}

\newtheorem{corollary}[theorem]{Corollary}
\newtheorem{lemma}[theorem]{Lemma}

\newcommand{\miika}[1]{\todo[inline]{Miika: #1}}
\newcommand{\teemu}[1]{\todo[inline, color=blue!20]{Teemu: #1}}

\newcommand{\jonni}[1]{\todo[inline, color=red!20]{Jonni: #1}} 

\usepackage[mathlines]{lineno}
\usepackage{etoolbox}

\newcommand*\linenomathpatch[1]{%
  \cspreto{#1}{\linenomath}%
  \cspreto{#1*}{\linenomath}%
  \csappto{end#1}{\endlinenomath}%
  \csappto{end#1*}{\endlinenomath}%
}
\newcommand*\linenomathpatchAMS[1]{%
  \cspreto{#1}{\linenomathAMS}%
  \cspreto{#1*}{\linenomathAMS}%
  \csappto{end#1}{\endlinenomath}%
  \csappto{end#1*}{\endlinenomath}%
}

\expandafter\ifx\linenomath\linenomathWithnumbers
  \let\linenomathAMS\linenomathWithnumbers
  \patchcmd\linenomathAMS{\advance\postdisplaypenalty\linenopenalty}{}{}{}
\else
  \let\linenomathAMS\linenomathNonumbers
\fi

\linenomathpatch{equation}
\linenomathpatchAMS{gather}
\linenomathpatchAMS{multline}
\linenomathpatchAMS{align}
\linenomathpatchAMS{alignat}
\linenomathpatchAMS{flalign}

\makeatletter
\patchcmd{\mmeasure@}{\measuring@true}{
  \measuring@true
  \ifnum-\linenopenaltypar>\interdisplaylinepenalty
    \advance\interdisplaylinepenalty-\linenopenalty
  \fi
  }{}{}
\makeatother

\ifCLASSINFOpdf
\else
\fi
\newcommand*{\ex}{\exists \mathbb{R}}
\newcommand*{\exexp}{\ex_{\exp}}
\newcommand*{\exexps}{\ex_{\exp}^{<}}
\newcommand{\ETR}{\operatorname{ETR}}
\newcommand{\FEASFOUR}{\operatorname{4-FEAS}}
\newcommand{\FEASFOURtau}{\operatorname{4-FEAS_\tau}}
\newcommand{\ECF}{\operatorname{ECF}}

\begin{document}
%
\title{\hspace{-5mm}Complexity of Neural Network Training and ETR:\\
\hspace{-5mm}Extensions with Effectively Continuous Functions
}

\author{
\IEEEauthorblockN{Teemu Hankala,\\ Miika Hannula, \\and Juha Kontinen}
\IEEEauthorblockA{University of Helsinki\\
 Helsinki, Finland\\
Email: \{teemu.hankala,miika.hannula,juha.kontinen\}@helsinki.fi
}
\and
\IEEEauthorblockN{Jonni Virtema}
\IEEEauthorblockA{University of Sheffield\\ Sheffield, UK\\
University of Helsinki\\ Helsinki, Finland\\
Email: j.t.virtema@sheffield.ac.uk}
}

%


\maketitle

\begin{abstract}
We study the complexity of the problem of training neural networks defined via various activation functions. The training problem is known to be $\exists\mathbb{R}$-complete with respect to linear activation functions and the ReLU activation function. We consider the complexity of the problem with respect to the sigmoid activation function and other effectively continuous functions. We show that these training problems are polynomial-time many-one bireducible to the existential theory of the reals extended with the corresponding activation functions. In particular, we establish that the sigmoid activation function leads to the existential theory of the reals with the exponential function. It is thus open, and equivalent with the decidability of the existential theory of the reals with the exponential function, whether training neural networks using the sigmoid activation function is algorithmically solvable. In contrast, we obtain that the training problem is undecidable if sinusoidal activation functions are considered. Finally, we obtain general upper bounds for the complexity of the training problem in the form of low levels of the arithmetical hierarchy.
\end{abstract}


%
\IEEEpeerreviewmaketitle

\section{Introduction}

We study the computational complexity of neural network (NN) training problems parameterized by various activation functions. Such a training problem asks,
given a neural network architecture, finite training data, a cost function and a threshold, whether there exist real-valued edge weights and neuron biases for the network such that the total training error with respect to the training data is below the given threshold value. Neural network training problems have been studied extensively (see \cite{Jiri02} for a comprehensive survey  of the history) and recently a version of this problem was shown to be complete for the complexity class $\ex$ \cite{AbrahamsenKM21,Bertschinger22}. The aforementioned complexity class  $\ex$  is defined as the class of decision problems $P$ that are polynomial-time reducible to the
existential theory of the reals, that is, to the corresponding decision problem $\ETR$ \cite{SchaeferS17}.
In other words, there is a polynomial-time computable function $f$ such that $x\in P$ if and only if $f(x)$ is an existential first-order ($\FO$) formula that is true over the structure $(\mathbb{R}, +,\times,<,=, 0,1)$ of the ordered field of the real numbers. The latter question is easily seen to be equivalent to deciding whether a finite system of polynomial equations and inequations with integer coefficients and real unknowns has a solution. 

The class $\exists\mathbb{R}$ has turned out to have many interesting complete problems, for instance the so-called art gallery problem \cite{AbrahamsenAM22},  two-dimensional packing problems \cite{AbrahamsenMS20}, and certain decision problems on symmetric Nash Equilibria  \cite{BiloM17}. The activation functions in the $\ex$-complete training problem mentioned above are restricted to linear functions and the rectified linear unit (ReLU) activation function. In practice, there are also other useful activation functions such as the standard logistic sigmoid function, which is defined in terms of the exponential function, and sinusoidal activation functions.
However, whereas $\ETR$ is $\NP$-hard and included in $\PSPACE$ \cite{Canny88},
it is a long-standing and influential open problem posed
by Alfred Tarski whether exponential arithmetic is decidable. This problem is known to be related to another major open question -- \emph{Schanuel's conjecture} -- in transcendental number theory (see, e.g., \cite{Wilkie1997,Servi08}). On the other hand, the theory of the reals extended by the sine function
is known to be undecidable, even under some further restrictions on the allowed syntax of the formulae \cite{Richardson68}.

In this article, we generalize the connection between $\ETR$ and the neural network training problem from linear functions and the ReLU activation function to arbitrary
real-valued functions.
On the $\ETR$ side, the corresponding activation function is added as a new function symbol to the signature of the underlying structure   $(\mathbb{R}, +,\times,<,=, 0,1)$. We show that
the NN training problem using $f$ as an activation function together with the identity activation function is complete for the extension $\ex_f$ of the class $\ex$. In other words,  this training problem is polynomial-time many-one bireducible to
the existential theory of the reals extended with $f$.
Besides a single activation function $f$, we allow the use of a set $\tau$ of activation functions
both in the training problem and in the definition of complexity classes of the form $\ex_\tau$.
%
Our result implies that the decidability of the training problem of neural networks using the sigmoid activation function is equivalent to the decidability of (the existential fragment of) exponential arithmetic. In fact, due to a certain model theoretic property (i.e., \emph{model completeness}) of the theory of exponential arithmetic, this theory is decidable if and only if its existential fragment is decidable, and for showing  decidability,  it suffices to show it to be recursively enumerable ($\Sigma^0_1$) 
(see, e.g., \cite{Servi08} for further discussion).  Another immediate consequence is that  the neural network training problem is undecidable if the sine function is allowed to be used for neural activation. Interestingly, sinusoidal activation functions have already been observed to be hard to train in simulations, and have been analyzed, for instance, in \cite{Lapedes}.

Moreover, we establish upper bounds for the complexity of the NN training problem in the case of \emph{effectively continuous} activation functions.
The class of effectively continuous functions (see \cite{Servi08}) contains, for example, the sigmoid function and sinusoidal activation functions.
We show, using rational approximations and basic topological properties, that it is possible to place these problems and the theories $\ex_\tau$ into $\Sigma^0_3$, the third level of the arithmetical hierarchy.
We also show that for effectively continuous functions, the complexity of $\ex_\tau$ drops to  $\Sigma^0_1$ if the use of the equality sign is disallowed in the formulae.
However, it is not known whether this seeming distinction is real or not, for instance, in the case of the
exponential function and $\exexp$. As a matter of fact, this question can be shown to be equivalent to
Tarski's exponential function problem.
Note that disallowing equalities in favour of the strict order relation has no implications in the complexity of $\ETR$, as has been shown in \cite{SchaeferS17}.
In addition, the neural network training problem can be restricted in a manner which
is contained in the corresponding strictly ordered variant of the class $\ex_\tau$.
Under this formulation, the set of weights and biases satisfying an instance of the training problem can be
observed to be \emph{stable} in the sense that it is an open set in the standard Euclidean topology.

In the next two sections, we will present model-theoretic definitions and properties
of the complexity classes that are obtained by extending $\ETR$ with new real-valued functions.
We give an emphasis on effectively continuous functions and give a short summary of their basic properties, restated from \cite{Servi08}.
Next, we formulate such variants of the NN training problem parameterized by different
activation functions that are complete for the corresponding extensions of the complexity class $\ex$.
In the last section, we formulate some shortcomings of our neural network construction as open questions.

\section{Preliminaries}
For real numbers $a$ and $b$, we let $(a,b)$ and $[a,b]$ denote the open and closed intervals with endpoints $a$ and $b$, respectively. We write $\lvert a \rvert$ for the absolute value of $a$. For sets $A$ and $B$, we write $A^B$ to denote the set of all functions with domain $B$ and co-domain $A$.

We assume familiarity with basic complexity classes such as $\PTIME$, $\NP$ and $\PSPACE$, and briefly review some concepts from computability theory (see, e.g., \cite{Arora09}).
In this context, it is customary to assume that finite objects are encoded as natural numbers. A relation of natural numbers $S \subseteq \mathbb{N}^n$, for a fixed $n\geq 1$, is \emph{computable} (or \emph{decidable}) if there exists a Turing machine that given any tuple $x \in \mathbb{N}^n$ decides whether or not $x \in S$. The \emph{arithmetical hierarchy} is formed by the classes of sets $\Sigma^0_k$ and $\Pi^0_k$, $k \in \mathbb{N}$, defined recursively as follows: 
\begin{itemize}
\item The class $\Sigma^0_0=\Pi^0_0$ consists of all computable relations $S$ over $\mathbb{N}$.
\item The class $\Sigma^0_{k+1}$ consists of all relations $S(x_1, \dots, x_n)$ that are definable by a sentence of the form \\
$ \exists y_1 \dots y_m R(x_1, \dots, x_n, y_1, \dots, y_m)$, where $R$ is a relation in $\Pi^0_k$.
\item The class $\Pi^0_{k+1}$ consists of all relations $S(x_1, \dots, x_n)$ that are definable by a sentence of the form \\
$ \forall y_1 \dots y_m R(x_1, \dots, x_n, y_1, \dots, y_m)$, where $R$ is a relation in $\Sigma^0_k$.
\end{itemize}
The classes $\Sigma^0_1$ and $\Pi^0_1$ thus form the sets of \emph{recursively enumerable} relations and \emph{co-recursively enumerable} relations, respectively. In addition, a well-known consequence of Post's theorem is that for each natural number $k$ the strict inclusions
$\Sigma^0_k \subsetneq \Sigma^0_{k+1}$ and $\Pi^0_k \subsetneq \Pi^0_{k+1}$ hold.





The first-order language of \emph{real arithmetic}, written $\FO(+,\times,<,=,0,1)$,
is given by the grammar
\begin{equation}\label{eq:synt}
\phi \ddnf i<i \mid i= i \mid \phi \land \phi \mid \phi \lor \phi \mid \exists x \phi \mid \forall x \phi\text{,}
\end{equation}
where $i$ stands for numerical terms given by the grammar
\[
i \ddnf 0 \mid 1 \mid x \mid i \times i \mid i+i\text{,} 
\quad \text{where $x$ is a first-order variable.}
\]
Note that adding negation to \eqref{eq:synt} does not increase the expressiveness of the language. A negated formula can be expressed positively via a negation normal form transformation and subsequent positive rewriting of negated atomic formulae.


Furthermore, we consider various extensions of the first-order language of real arithmetic. For a collection of relation and function symbols $\calC$, we write $\FO(\calC)$ to denote the variant of $\FO(+,\times,<,=,0,1)$ that uses only function symbols and relation symbols in $\calC$. 
\emph{Existential real arithmetic} $\EFO(+,\times,<,=,0,1)$ is obtained from \eqref{eq:synt} by dropping universal quantification. The logic
$\EFO(\calC)$ is defined analogously.

A formula $\phi$ is \emph{open} if some variable appears free in $\phi$, and otherwise \emph{closed}. Closed formulae are referred to as \emph{sentences}. Given a first-order structure $\A$ and a variable assignment $s$, we write $\A \models_s \phi$ in the case that $\phi$ is true in $\A$ with respect to $s$. If $\phi$ is a sentence, we write $\A \models \phi$ if $\phi$ is true in $\A$.

The semantics for the language of real arithmetic is defined over the fixed structure $(\RE, +, \times, <, =, 0, 1)$ of real arithmetic in the usual way, and similarly for $\FO(\calC)$ where the symbols in $\calC$ have their usual interpretations.
As a slight abuse of notation, we use $f$ to denote both a real-valued function
and the corresponding function symbol.
Moreover, we allow function symbols to be interpreted as partial functions. Any atomic formula
with some undefined term with respect to an assignment $s$ is defined to be false.

We use shorthands $\RE$ and $\REs$ for the models $(\RE, +, \times, <, =, 0, 1)$ and $(\RE, +, \times, <, 0, 1)$, respectively. For a collection $\tau$ of additional functions,
we write $\RE_\tau$ and $\REs_\tau$ for the corresponding extension of
the models $\RE$ and $\REs$, respectively.
%
We define the complexity class
$\ERE_\tau$
as the collection of decision problems that have a polynomial-time many-one reduction to closed formulae of 
the existential real arithmetic with additional functions from $\tau$.
We also consider a subclass of
$\ERE_\tau$
defined in terms of strict inequalities: the class
$\EREs_\tau$
is defined otherwise as
$\ERE_\tau$,
except that we drop equality atoms $i = i$ from \eqref{eq:synt}.


\section{The Existential Theory of the Reals in the Arithmetical Hierarchy}

The full first-order theory of the ordered field of the real numbers admits an efficient version
of quantifier elimination and thus is decidable, a result first shown by Alfred Tarski.
Furthermore, the satisfiability problem of the existential fragment of this theory is known to be in $\PSPACE$ \cite{Canny88}. In general, though, it is not known if the corresponding
existential theory is decidable
when new function symbols are added to the underlying model.
The addition of the sine function is known to lead to undecidability, whereas Tarski's
exponential function problem has been famously open since the 1950s.


In this section we show that if the added functions are assumed
to satisfy certain fairly general computable topological properties,
the satisfiability problem of the resulting existential theories can be bounded from above
using low levels of the arithmetical hierarchy.
Specifically, to this end, we use the following definition of effectively continuous
functions from \cite{Servi08}.


\begin{definition}[{\cite[Definition 4.2.1]{Servi08}}]
A function $f \colon \mathbb{R}^n \to \mathbb{R}$ is called \emph{effectively continuous} if there exists a computable function $g$ with the following properties:
\begin{itemize}
\item
To each $n$-tuple of open intervals $(a_1, b_1),\, \dots,\, (a_n, b_n)$ with rational endpoints,
the function $g$ associates an open interval $(c, d)$ with rational endpoints such that
\[
\forall x \left(x\in(a_1,b_1)\times \dots \times (a_n,b_n)\implies f(x)\in (c,d)\right)\text{.}
\]
\item
$\forall M>0 \,\, \forall \varepsilon >0 \,\, \exists \delta >0 \,\, \forall a_i,b_i$ it holds that, if $c,d$ is the output of $g$ on input $a_i,b_i$, then 
\[
\bigwedge_{i\in [n]}(a_i,b_i) \subseteq [-M,M]\land   \lvert a_i-b_i\rvert <\delta\implies \lvert c-d\vert <\varepsilon.
\]
\end{itemize}
In this case, it is said that $g$ \emph{computes} the function $f$. In addition,
the set of all effectively continuous functions is denoted by $\ECF$.
\end{definition}




As noted in \cite{Servi08}, there is at most one effectively continuous function for each computable function. Thus the set $\ECF$ is countable, and not every constant function is
effectively continuous. On the other hand, rational constant functions,
the identity function $x \mapsto x$, the absolute value function $x \mapsto |x|$, together with addition, multiplication, division, and the exponential function
and basic trigonometric functions are effectively continuous. Using these functions and the following lemma,
all rational functions and
many neural activation functions, such as the ReLU function and
the sigmoid function $\sigma \colon x \mapsto 1 / (1 + \exp(-x))$, can be seen to be in $\ECF$.


\begin{lemma}[{\cite[Lemma 4.2.6]{Servi08}}]\label{comp}
 The set $\ECF$ of all effectively continuous functions is effectively closed under composition, i.e., given computable functions $g$ and $g'$ which compute two effectively continuous functions $f$ and $f'$, respectively, we can effectively find a computable function $g''$ which computes the composition $f \circ f'$ (when the latter is defined).

\end{lemma}


The following two lemmas state some basic topological properties of effectively continuous functions
that are central in our approach.

\begin{lemma}[{\cite[Remark 4.2.2 (rephrased)]{Servi08}}]\label{lem:ecfcont}
Effectively continuous functions are uniformly continuous
on every compact cube of the form $[-M, M]^n$, where $M > 0$.
In particular, effectively continuous functions are continuous in the Euclidean topology.
\end{lemma}

\begin{lemma}[{\cite[Proposition 4.2.10 (rephrased)]{Servi08}}]
\label{strictorderreatoms}
Let $T$ be a set of effectively continuous functions. Then the set 
\[ \left\{(g,\bar{x} ): \ \bar{x} \in \mathbb{Q}^n,\ g \in T \textrm{ and }g(\bar{x})<0\right\}\]
is recursively enumerable relatively to $T$, i.e., there is an algorithm which, on input $\bar{x} \in \mathbb{Q}^n$, and a code for $g\in T$, stops if and only if $g(\bar{x})<0$.
\end{lemma}
This lemma implies that the set of rational solutions $\bar{x} \in \Q ^ n$ for $f(\bar{x}) < 0$
is recursively enumerable when $f$ is effectively continuous.



We utilise some rudimentary notions and knowledge from
point-set topology summarised in the following lemma (for basics of point-set topology see, e.g., the monograph \cite{willard2004general}).
In order to simplify the notation, for a topological space $X$, we use $X$ to denote also the underlying set of the space.
A set $D\subseteq A$ is said to be a \emph{dense subset of $A$} (in some topology) if the set $D$ intersects every non-empty open subset of $A$. A topological space $X$ is \emph{compact} (more specifically, $\omega$-compact) if every countable collection $E$ of
open subsets of $X$ such that $X \subseteq \bigcup E$ has a finite subset $F \subseteq D$ so that the condition $X \subseteq \bigcup F$ holds.
\begin{lemma}[Topology toolbox]\label{lem:topology}
	Let $X$ and $Y$ be topological spaces, $f\colon X\rightarrow Y$ a continuous function, $S$ a set in $X$, $A$ and $B$ open sets in $X$, and $C$ an open set in $Y$. Then the following claims hold.
	\begin{itemize}
		\item The subsets $A\cap B$, $A\cup B$, and $f^{-1}[C]$ are open in $X$.
		\item The product $A^k$ is open in the product space $X^k$.
		\item If $D$ is a dense subset of $S$ then $D^k$ is a dense subset of $S^k$ in the corresponding product topology.
        \item If $S^k$ is open in $X^k$, then any projection of $S^k$ is open in the corresponding product space $X^n$.
        \item The product topology of compact spaces is compact.
        \item If $X$ is a compact metric space, then $X$ is sequentially compact, that is, if $(x_i)_{i \in \N}$, is a sequence of elements of $X$, then there is a converging subsequence of $(x_i)_{i \in \N}$.
	\end{itemize}
\end{lemma}
The following lemma is an easy consequence of the previous topological facts.
\begin{lemma} \label{openformula}
Let $\A$ be a model with a vocabulary $\tau$, and $\T$ a topology of the
universe $A$ of $\A$ such that
\begin{itemize}
\item for every $n$-ary relation symbol $R$ in $\tau$, the interpretation
$R^\A$ is an open set in the corresponding product topology of
$A^n$ obtained from $\T$, and
\item for every $n$-ary function symbol $f$ in $\tau$, the interpretation
$f^\A$ is continuous with respect to the product topology of $A^n$ obtained
from $\T$.
\end{itemize}
Let $\phi(\bar{x})\in \EFO(\tau)$
be a formula
with free variables in the $n$-tuple $\bar{x}$.
Then the set of $n$-tuples $\bar{a}$ satisfying
$\phi(\bar{a})$ in $\A$ is an open set in the product topology of $A^n$ based
on $\T$.

%

\begin{proof}
The proof proceeds by induction; we utilise topological facts from Lemma \ref{lem:topology}. First note that the interpretations of terms are continuous functions, since by induction they can be computed by composing continuous functions.

The interpretations of atomic formulae are open, for pre-images of open sets of continuous functions are open.
Finally, if the interpretation of $\phi$ and $\psi$ are open, the interpretations of $\phi\land \psi$ and $\phi\lor \psi$ are open, for the intersections and unions of open sets are open, and the interpretation of $\exists x \phi$ is open, for projections of open sets are open.
%
\end{proof}

\end{lemma}

\begin{lemma} \label{denseformula}
Let $\A$, $\tau$ and $\T$ be as in the previous lemma. Let $D$ be a dense subset of the universe
$A$ in the topology $\T$, and let $\phi(\bar{x})\in\EFO(\tau)$ be a formula with free variables in $\bar{x}$. Then $\phi$ is satisfiable in $\A$ if and only if $\phi(\bar{a})$ is true in $\A$ for some tuple $\bar{a}$ from the dense set $D$.
%
\end{lemma}

\begin{proof}
The $\Leftarrow$ direction is immediate. For the converse, let $\phi$ be an $\EFO(\tau)$ formula with $k$ free variables that is satisfied in $\A$. By Lemma \ref{openformula}, $\phi$ defines an open set in $A^k$. Note that, since $D$ is a dense subset of $A$ in the topology $\T$, $D^k$ is a dense subset of $A^k$ in the corresponding product topology. Now by the definition of denseness, the intersection of $D^k$ and the set defined by $\phi$ is non-empty, and thus $\phi(\bar{a})$ is true in $\A$ for some tuple $\bar{a}\in D^k$.
%
%
%
%
%
\end{proof}

By the density of rational numbers, we can now obtain an upper bound for the satisfiability problem
of formulae when function symbols are interpreted using effectively continuous functions.
First, if the equality sign is not allowed in the formulae, the remaining strict order relation
corresponds to open sets in the Euclidean topology.

\begin{theorem} \label{strictformula}
Let  $\phi(\bar{x})$ be an existential formula  without identity over the structure $\RE_{\ECF}$. 
Then the set of rational solutions $\bar{a} \in \Q^n$
for $\RE_{\ECF} \models \phi(\bar{a})$ is recursively enumerable.
Moreover, the complexity class $\EREs_{\ECF}$ is contained in $\Sigma^0_1$.
%
\end{theorem}

\begin{proof}
The proof proceeds by induction of the structure of  $\EFO(\{<\}\cup \ECF)$-formulae.
First note that each atomic formula can be equivalently expressed by an inequality of the form $f(\bar{y}) < 0$ for some effectively
continuous function $f$ and tuple of variables $\bar{y}$. Thus, by Lemma \ref{strictorderreatoms}, rational solution sets for atomic formulae are recursively
enumerable. The cases for conjunction and disjunction are immediate, for finite intersections and unions of recursively enumerable sets are recursively enumerable.

Finally, consider the case for $\exists y \psi(\bar{x},y)$. By the induction hypothesis, the rational truth set of $\psi(\bar{x},y)$ is recursively enumerable; we claim that its projection $P$ to the variables in $\bar{x}$ is the recursively enumerable rational truth set of $\exists y \psi(\bar{x},y)$.
Note first that projections of recursively enumerable sets are recursively enumerable, 
and thus it suffices to show that $P$ is the rational truth set of $\exists y \psi(\bar{x},y)$. Assume that $\phi(\bar{a},b)$, where $\bar{a}$ is a tuple of rational numbers and $b$ is a real number, is true in $\REs_{\ECF}$. Now since, by Lemma \ref{openformula}, the truth set of $\phi(\bar{x},y)$ is open, there exists a rational number $c$ such that  $\phi(\bar{a},c)$ is true in $\REs_{\ECF}$, from which the claim follows.

Now $\EREs_{\ECF}$ is in $\Sigma^0_1$, since by Lemma \ref{denseformula}, non-emptiness of truth sets of $\EFO(\{<\}\cup \ECF)$-formulae are equivalent to checking non-emptiness of their rational truth sets, and $\Sigma^0_1$ is clearly closed under polynomial-time reductions.
\end{proof}
In the following, a \emph{$\tau$-polynomial} is a $\tau$-term where, except for
the arithmetic operations $+$ and $\times$, functions $g$ can only appear with variable or constant arguments, that is, in the form $g(x)$ or $g(c)$. Furthermore, the (total) degree of a polynomial $P$ is defined as the maximum over the sum of the variable exponents in each monomial term occurring in $P$. 
Note that the minus sign can be avoided by shuffling monomial terms
from one side of the equations and inequations to the other.
\begin{lemma}\label{lem:nash}
Let $\tau$ be a set of function symbols and $\phi(\bar{x})\in \EFO(\{+,\times,<,=,0,1\}\cup \tau)$. Then we can construct in polynomial time with respect to the length of the formula $\phi(\bar{x})$ a $\tau$-polynomial $P(\bar{x},\bar{y})$ of degree at most $4$
such that $\phi(\bar{x})$  is equivalent to 
\[ \exists \bar{y}P(\bar{x},\bar{y})=0 \]
over the expanded reals $\RE_{\tau}$.
\end{lemma}
\begin{proof} The proof of this lemma is analogous to the proof of \cite[Lemma 3.2]{SchaeferS17}. In fact, it suffices to first simplify all subterms of $\phi$ of the form $g(t)$ to the required form using new variables: replace $g(t)$ by $g(y)$ and add a new conjunct $y=t$. Then the translation in \cite{SchaeferS17} can be used directly by treating terms $g(y)$ exactly the same way as in the original proof.
\end{proof}

Analogously to the $\ex$-complete decision problem $\FEASFOUR$ in \cite{SchaeferS17},
we define $\FEASFOURtau$ to be the decision problem that asks, given a $\tau$-polynomial of degree at most $4$, if the polynomial is feasible, i.e., has
a root in the extended model $\RE_\tau$. Using the previous lemma, we obtain a complete problem for
the corresponding extension of $\ex$ with the functions in $\tau$,
thus justifying the use of the notation $\ex_\tau$.

\begin{corollary}
$\FEASFOURtau$ is $\ex_\tau$-complete.
\end{corollary}

Note that for product spaces of the form $\RE^n$, the standard Euclidean topology coincides with
the product topology and is sequentially compact. With the help of Lemma \ref{lem:nash},
we can now give an upper bound for complexity class of the existential theory of the
reals extended with effectively continuous functions.

\begin{theorem}
\label{EREECFSigma3}
$\ERE_{\ECF}$ is contained in $\Sigma^0_3$.
%
%
\end{theorem}
\begin{proof}
%
Let $\phi$ be an existential sentence 
over $\RE_{\ECF}$.
%
By Lemma \ref{lem:nash}, we may assume that $\phi$ is of the 
form  $\exists \bar{x} P_\tau(\bar{x}) = 0$, where $P_\tau$ is a $\tau$-polynomial over a vocabulary $\tau$ that contains function symbols for the effectively continuous functions.
By Lemmas \ref{comp} and \ref{lem:ecfcont}, the function defined by $P_\tau$ is effectively continuous.
%
We observe that $\exists \bar{x} P_\tau(\bar{x}) = 0$ is equivalent over $\RE_{\ECF}$ to the expression
\begin{equation}\label{eq:sigma3formula}
\exists d> 0 \,\, \forall \varepsilon > 0 \,\, \exists \bar{x}\in {B}(\bar{0},d): |P_\tau(\bar{x})|< \varepsilon\text{.}
\end{equation}
For this claim, $\exists \bar{x} P_\tau(\bar{x}) = 0$ clearly implies 
\eqref{eq:sigma3formula}. For the converse direction, \eqref{eq:sigma3formula} entails that there is an infinite sequence of tuples $(\bar{x}_i)_{i \in \N}$ inside some open ball ${B}(\bar{0},d)$ such that $\lim_{i \to \infty} P_\tau(\bar{x}_i) = 0$. Since the closure $\bar{B}(\bar{0},d)$ is compact, the sequence $(\bar{x}_i)_{i \in \N}$ has a subsequence $(\bar{y}_{j})_{j \in \N}$ that converges to a limit point $y\in \bar{B}(\bar{0},d)$. Since $P_\tau$ is continuous, and $\lim_{j \to \infty} P_\tau(\bar{y}_j) = 0$, we obtain that $P_\tau(\bar{y})=0$ and, in particular, that the claim $\exists \bar{x} P_\tau(\bar{x}) = 0$ holds.

The parameters $\varepsilon$ and $d$ can without loss of generality be assumed to be rational in \eqref{eq:sigma3formula}. The innermost 
existentially quantified part of \eqref{eq:sigma3formula} then describes a predicate that is definable by an existential formula with $\varepsilon$ and $d$ as rational constant parameters over $\RE^<_{\ECF}$ and thus is in
 $\Sigma^0_1$ by Theorem \ref{strictformula}. The full expression \eqref{eq:sigma3formula} hence defines a $\Sigma^0_3$ predicate.
Since $\ERE_{\ECF}$ is defined as the class of problems that can be reduced in polynomial time to sentences of the form of $\phi$, we conclude that $\ERE_{\ECF}$ is contained in $\Sigma^0_3$.
\end{proof}

Now the following results are corollaries of the preceding theorems.

\begin{theorem}
$\exexps \subseteq \Sigma^0_1$ and $\exexp \subseteq \Sigma^0_3$.
\end{theorem}
In fact, by the model completeness of the first-order theory of
$\RE_{\exp}$ (see \cite{Wilkie1997, Servi08}), the decidability
of its existential fragment is equivalent to it being recursively enumerable,
and also equivalent to the full
first-order theory of $\RE_{\exp}$ being decidable. Thus, the precise relationship of $\exexp$
to $\Sigma^0_1$ and $\Sigma^0_3$ is an open problem.
These observations can be
generalized to any $\RE_\tau$ with a model-complete first-order theory and
to any set $\tau$ of effectively continuous functions.

Below, for convenience of stating the theorem, we assume below that  $\ex_{\sin}$ and $ \ex^<_{\sin}$ are closed under computable reductions.
\begin{theorem}\label{sincomplete}

 $\Sigma^0_1\subseteq \ex_{\sin} \subseteq \Sigma^0_3$ 
and $ \ex^<_{\sin}=\Sigma^0_1$.
\end{theorem}
\begin{proof}
Note that the formula $\sin(y)=0 \land 4<y \land y<7$ defines $2\pi$ in $y$, and that the formula
\[\sin(2\pi \times x)=0 \land (x = 0 \lor x > 0)\] defines natural numbers over the reals,
with the help of the previously defined $2\pi$. Hence $\Sigma^0_1$-hardness follows from the Davis--Putnam--Robinson--Matiyasevich theorem.
For the second item, hardness follows from the results of  \cite{Laczkovich}.
\end{proof}

\section{The Neural Network Training Problem}

Under a suitable formulation, the complexity of the neural network training problem
is complete for the complexity class $\ERE_\tau$, where the set $\tau$
includes all those unary real-valued functions that are being used for neural activation.
In the case of effectively continuous activation functions,
the complexity of the training problem can be bounded using low levels of the
arithmetical hierarchy.
We concentrate on \emph{feedforward neural networks} under the following definitions.

\begin{definition}
A \emph{neural network architecture}
is a pair $N = (G, \varphi)$
over a finite directed acyclic graph $G = (V, E)$
and a function $\varphi$ with the following definitions and properties:
\begin{itemize}
\item The vertices of the graph are called \emph{neurons}.
A neuron of the network architecture $N$ is called
an \emph{input neuron} if there are no incoming edges to it. A neuron is an \emph{output neuron}
if there are no outgoing edges from it. Every other vertex of the graph
is called a \emph{hidden neuron}.
\item The function $\varphi$ maps each non-input neuron $v$ to a function $\varphi(v) \colon \RE \to \RE$.
The function $\varphi(v)$ is the \emph{activation function} of the neuron $v$.
The value $\varphi(v)$ is alternatively denoted by $\varphi_v$.
\end{itemize}
%
%
%
%
\end{definition}

\begin{definition}
\label{neuralnetwork}
A \emph{neural network} is a triple $\NN = (N, w, b)$ over a neural network
architecture $N$ together with real-valued functions $w$ and $b$ with the following properties:
\begin{itemize}
\item The function $w$ maps each edge $e$ of the underlying graph to a real number. The value
$w(e)$ is the \emph{edge weight} of $e$ in the neural network and is denoted by $w_e$.
\item The function $b$ maps each non-input neuron $v$ to a real number, called the \emph{bias} of $v$,
and denoted in short by $b_v$.
\end{itemize}
\end{definition}

\begin{definition}
A \emph{data point} for a neural network architecture $N$ is a function $d$
that maps a real number to each input neuron and to each output neuron.
%
\end{definition}


\begin{definition}
\label{def:neuralfunction}
Let $\NN = (N, w, b)$ be a neural network. The \emph{neural function} of
$\NN$ is the unique function $g$ with the following definitions and properties:
\begin{itemize}
\item The domain of $g$ is the set of
all possible data points for the network architecture $N$. For each data point $d$,
the value $g(d)$ is a function mapping each neuron to a real number, called the \emph{neural value}
computed by the network for the given neuron.
This function is also referred to as $g_d$.


\item For each data point $d \in \dom(g)$ and for each input neuron $v$ it
holds that $g_d(v) = d(v)$.

\item For each $x \in \dom(g)$ and for each neuron $v \in V$ that is not an input neuron
it holds that
\begin{equation}
\label{eq:neuralactivation}
g_d(v) = \varphi_v\left(b_v + \sum_{u \in P} \left(w_{(u, v)} \times
g_d(u)\right)\right)\text{,}
\end{equation}
where $P$ is the set of predecessors of $v$ in the underlying directed acyclic graph.




\end{itemize}

\end{definition}

In other words, the neural function for a fixed data point $d$ is defined
recursively as follows: For each input neuron $v$,
the value given by the neural function for $v$ is the value $d(v)$ as
given by the data point. For every other neuron $v$, the neural value $g_d(v)$
is computed by adding the bias value $b_v$ to the sum of the incoming neural values,
each multiplied by the corresponding edge weight, and finally this sum is mapped through
the activation function $\varphi_v$.
The neural function is well-defined, since the network architecture
is specifically assumed to be acyclic.

%
%


\newcommand{\NNtauTRAINING}{\operatorname{NN_\tau-TRAINING}}
\newcommand{\NNfTRAINING}{\operatorname{NN_f-TRAINING}}
\newcommand{\NNfsingletonTRAINING}{\operatorname{NN_{\{f\}}-TRAINING}}
\newcommand{\NNsinTRAINING}{\operatorname{NN_{\sin}-TRAINING}}
\newcommand{\NNsigmoidTRAINING}{\operatorname{NN_\sigma-TRAINING}}

In the following definition, to simplify notation,
we denote by $d_O$ the tuple consisting of the values of a data point
$d$ for the output neurons of the network architecture, where some fixed order on the output
neurons independent of the data points is assumed.

\begin{definition}
\label{def:nntraining}
Let $\tau$ be a set of unary functions. An \emph{$\NNtauTRAINING$ instance} is a tuple
of the form $(N, A_E, A_V, D, c, \prec, \delta)$ for a neural
network architecture $N = (G, \varphi)$ with the following properties:
\begin{itemize}
\item For each neuron $v$ of the network architecture $N$,
the activation function $\varphi_v$ is either the identity function $x \mapsto x$
or some function $f \in \tau$.
\item A subset $A_E$ of the edges of $N$ are marked as \emph{active edges} in the
training process. Similarly, $A_V$
is a subset of \emph{active neurons}. 
\item A finite set $D$ of data points that maps neurons to rational numbers.
\item A \emph{cost function} $c \colon \RE ^ {2m} \to \RE$, where $m$
is the number of output neurons of $N$, is assumed to be definable with an arithmetic expression.
\item The symbol $\prec$ is one the of relation symbols in the set $\{=, \leq, <\}$.
\item A \emph{threshold} value $\delta$ for allowed total training error is given as a rational number.
\end{itemize}

A pair $(w, b)$ is a \emph{satisfying solution}
for the training instance $(N, A_E, A_V, D, c, \prec, \delta)$,
if $(N, w, b)$ is a neural network as in Definition \ref{neuralnetwork}
so that for each edge $e$ of $G$ not in $A_E$ it holds that $w(e) = 1$ and
for each neuron $v$ not in $A_V$, $b(v) = 0$, and that the neural function $g$
satisfies the condition
\begin{equation}\label{eq:totalerror}
\sum_{d \in D} c\left((g_d)_O, d_O\right) \prec \delta\text{.}
\end{equation}


Then $\NNtauTRAINING$ is the decision problem that asks, given an $\NNtauTRAINING$ instance,
if there exists some satisfying solution for it. Furthermore, $\NNtauTRAINING^<$ is defined
otherwise as the decision problem $\NNtauTRAINING$, except that the strict order relation $<$
is the only allowed choice for the relation symbol $\prec$.
Furthermore, for a single function $f$, we use the notation $\NNfTRAINING$ for $\NNfsingletonTRAINING$,
and similarly for $\NNfTRAINING^<$.

\end{definition}

If $c$ is such a cost function that for all $m$-tuples $\bar{a}$ and $\bar{b}$
it holds that $c(\bar{a}, \bar{b}) = 0$ if and only if $\bar{a} = \bar{b}$,
then it is called \emph{faithful}. If $c$ is both faithful and non-negative,
threshold value $\delta = 0$ is used, and $\prec$ is replaced with the equality sign,
then \ref{eq:totalerror} is equivalent to
all data points $d \in D$ satisfying the condition $(g_d)_O = d_O$, that is,
$g_d(v) = d(v)$ for every output neuron of the network.

%
%


\begin{lemma}\label{NNEREstau}
For any set $\tau$ of unary real-valued functions, the problem $\NNtauTRAINING^<$ is in
the complexity class $\EREs_{\tau \,\cup\, \{-, \div\}}$.
\end{lemma}

\begin{proof}
Let $(N, A_E, A_V, D, c, \prec, \delta)$ be an $\NNtauTRAINING$ instance. We show
how to construct a
formula $\phi(\bar{w}, \bar{b})$
of $\EFO(\{+,-,\times,\div,<,0,1\} \cup \tau)$
in polynomial time such that
there exists some satisfying solution for the given instance of the training problem if and only if
$\RE_\tau \models \exists \bar{w} \, \exists \bar{b} \, \phi(\bar{w}, \bar{b})$,
where $\bar{w}$ and $\bar{b}$ are finite sequences of variables that in some fixed order
encode the weights and biases
for the active edges and active neurons, respectively. We use $g$ to denote
the intended neural function for the variables in $\bar{w}$ and $\bar{b}$.
Let $m$ be the number of output neurons of $N$.

First we show, that for each data point $d \in D$ and for each neuron $v$ of the network,
there is a term $t_{(d, v)}$ which evaluates to $g_d(v)$ over $\RE_\tau$.
Namely, if $v$ is an input neuron, we can express $t_{(d, v)}$ succinctly as the corresponding input value
given by $d$, for instance,
in the form of a fraction of two binary expansions of integral values. If $v$ is a non-input neuron,
a suitable $t_{(d, v)}$ can be obtained as in expression \ref{eq:neuralactivation} of
Definition \ref{def:neuralfunction} using the terms $t_{(d, u)}$ of the predecessors $u$ of $v$ and
the sets $A_E$ and $A_V$.
Note that the underlying graph structure is given as a part of the input,
and that the graph is assumed to be acyclic.

According to Definition \ref{def:nntraining}, the cost function $c$ is expressible as an
arithmetic expression. Then, for each $d \in D$, there is some term $t'_d$ that
evaluates to $c((g_d)_O, d_O)$ over $\RE_\tau$.
Furthermore, the given rational threshold value $\delta$ can be expressed by some
constant term $t''_\delta$
Now, the relational formula \ref{eq:totalerror} of Definition \ref{def:nntraining}
is expressible as $\sum_{d \in D} t'_d < t''_\delta$,
thus yielding the claimed formula $\phi(\bar{w},\bar{b})$.
\end{proof}

Note that in the previous proof, $\phi(\bar{w}, \bar{b})$ is a single quantifier-free
relational atom.
It also defines the set of all satisfying solutions for the given training instance as
a subset of $\RE^{|A_E| + |A_V|}$.
If all the functions in $\tau$ are effectively continuous, this solution set is open in the Euclidean topology
by Lemma \ref{openformula}, and its intersection with $\Q^{|A_E| + |A_V|}$
is recursively enumerable by Theorem \ref{strictformula}.
What is more, if the complete first-order theory of the model $\RE_\tau$ is \emph{o-minimal}
(see, e.g., \cite{Servi08} for definition and discussion), then the set of satisfying
solutions has only a finite number of connected components.
In particular, the model $\RE_{\exp}$
is known to be o-minimal, whereas $\RE_{\sin}$ is not o-minimal.
However, for example in the case of the sigmoid activation function, it is an open
problem whether any positive lower bounds for the radii of open neighbourhoods
for solutions in the open solution set can be computed,
in general. By the model completeness of $\RE_{\exp}$, this question is connected
to Tarski's exponential function problem via the so-called \emph{first root conjecture}
(see, e.g., \cite{Wilkie1997, Servi08}).

\begin{lemma}\label{trainingInETRtau}
For any set $\tau$ of unary real-valued functions, the problem $\NNtauTRAINING$ is in $\ex_\tau$.
\end{lemma}
\begin{proof}
Let $(N, A_E, A_V, D, c, \prec, \delta)$ be an $\NNtauTRAINING$ instance
and let $\phi(\bar{w}, \bar{b})$ be as in the proof of Lemma \ref{NNEREstau}.
It is enough to show that there is some formula $\phi'(\bar{w}, \bar{b})$ of
$\EFO(\{+,\times,<,=,0,1\}\cup \tau)$ that is
equivalent to $\phi(\bar{w}, \bar{b})$ over the structure $\RE_{\tau \,\cup\, \{-, \div\}}$.
For this, note that if $\theta(x)$ is a formula with some open variable $x$,
then formulae of the form $\theta(-y)$ and $\theta(y / z)$ can be defined existentially as
$\exists x (\theta(x) \land x + y = 0)$ and $\exists x (\theta(x) \land x \times z = y)$.
Repeating these steps, the function symbols $-$ and $\div$ can be eliminated from
$\phi(\bar{w}, \bar{b})$.
\end{proof}

Notice that the proof of Lemma \ref{trainingInETRtau} and the membership of the training
problem in the class
$\ex_\tau$ can be generalized to the case
where all of the data points, activation functions, threshold value and cost function
are definable using existential
formulae of the alphabet $\{+,\times,<,=,0,1\} \cup \tau$. The set $\tau$
may also include other than unary functions.

\newcommand{\ETRtauFLAT}{\operatorname{ETR_\tau-INV-FLAT}}

Next we define a decision problem for existential formulae in a certain normal form.
The definition is an extension of the problem called ETR-INV in \cite{AbrahamsenAM22},
now allowing the use of unary function symbols.

\begin{definition}\label{def:ETRtauINVFLAT}

Let $\tau$ be a set of unary real-valued functions.
Then $\ETRtauFLAT$ is the following decision problem: Given a finite set
$\calC$ of \emph{constraints}, each of one of the forms
\begin{equation}\label{constraintlist}
x = 1\text{,}\quad x + y + z = 0\text{,}\quad x \times y + 1 = 0\text{,}\quad x + f(y) = 0\text{,}
\end{equation}
where $x$, $y$ and $z$ are first-order variable symbols and $f$ is some function symbol in $\tau$,
determine whether the system of constraints of $\calC$
is satisfiable over $\RE_\tau$ using a single variable assignment.

\end{definition}

The constraint types of \ref{constraintlist} are called \emph{unit} constraints,
\emph{addition} constraints, \emph{inversion} constraints and \emph{function}
constraints, respectively.

\begin{lemma}\label{ETRtauFLATcomplete}

The problem $\ETRtauFLAT$ is $\ex_\tau$-complete.

\end{lemma}

\begin{proof}
We show that the satisfiability problem of the existential theory of the
structure $\RE_\tau$ is reducible in polynomial time to $\ETRtauFLAT$;
the opposite direction is clear.
Let $\phi$ be a sentence of the logic $\EFO(\{+,\times,<,=,0,1\} \cup \tau)$.
Simplifying all functional subterms as in the proof of Lemma \ref{lem:nash}
and proceeding as in the
proof of Theorem 16 of \cite{AbrahamsenAM22},
the satisfiability of $\phi$ can be reduced in
polynomial time to a finite set $\calC'$ of constraints of the following forms:
\begin{equation*}
x = 1\text{,}\quad x + y = z\text{,}\quad x \times y = 1\text{,}\quad  x = f(y)\text{.}
\end{equation*}
For every choice of first-order variables $x$, $y$ and $z$, the following
observation is true for every variable assignment $s$:
\begin{equation*}
\RE_\tau \models_s x + y = z \iff 
\RE_\tau \models_s \exists u \exists v (v + v + v = 0 \land z + u + v = 0 \land x + y + u = 0)\text{.}
\end{equation*}
Similarly, for monomial terms $t$ and $t'$, it holds that
\begin{equation*}
\RE_\tau \models_s t = t' \iff 
\RE_\tau \models_s \exists u \exists v (v + v + v = 0 \land t + u + v = 0 \land t' + u + v = 0)\text{.}
\end{equation*}
In particular, each of the constraints in $\calC'$ can be expressed equivalently
using at most a constant number of
$\ETRtauFLAT$ constraints, producing an equivalent system of constraints.
%
%
\end{proof}



Depending on the functions in the set $\tau$,
some of the constraints in Definition \ref{def:ETRtauINVFLAT}
can be left out without losing $\ex_\tau$-completeness. For instance,
the exponential function can be used in order to switch between addition
and multiplication in existential formulae.

%
%
%

%
%
%
%
%

%
%
%

The proof of the following theorem is based
on the neural network construction of \cite{AbrahamsenKM21}, in which it was used to obtain
the first known version of an $\ex$-complete NN training problem.

\begin{theorem}\label{NNcomplete}

Let $\tau$ be a set of unary real-valued functions. Then the neural network training problem
$\NNtauTRAINING$ is $\ex_\tau$-complete.

\end{theorem}

\begin{proof}

By Lemmas \ref{trainingInETRtau} and \ref{ETRtauFLATcomplete} it is enough to show that
the problem $\ETRtauFLAT$ has a polynomial-time reduction
to $\NNtauTRAINING$.
To this end, let $\calC$ be a finite set of $\ETRtauFLAT$ constraints. We may assume
that $\calC$ does not include unsatisfiable constraints of the form $x \times x + 1 = 0$.
Let $W$ be
the set of all variable symbols that appear in any of the constraints of $\calC$.
Next, we describe how to construct an $\NNtauTRAINING$ instance
$(N, A_E, A_V, D, c, \prec, \delta)$ in polynomial time so that this instance has a satisfying solution
if and only if the system $\calC$ of constraints is satisfiable.

First,
let $(A_V, \prec, \delta) \dfn (\varnothing, =, 0)$. In particular, biases are assumed to be $0$
for all neurons.
Let $c$ be any faithful and non-negative cost function.
In the sequel,
we will refer to the variables of individual constraints
also using indices in the set $\{1, 2, 3\}$
according to
the list
\begin{equation}\label{cindex}
x_1 = 1\text{,}\quad x_1 + x_2 + x_3 = 0\text{,}\quad x_1 \times x_2 + 1 = 0,\quad
x_1 + f(x_2) = 0\text{.}
\end{equation}
Notice that these indices are based on the positions of the variables in the constraints
in such a manner that,
for example, in a constraint of the form $x + x + y = 0$, indices $1$ and $2$ would
refer to $x$ and $3$ to $y$.
For each $C \in \calC$, let $l_C$ be the appropriate number of these indices for $C$.

We define the underlying graph $G$ of the network architecture
$N = (G, \varphi)$, case by case, based on
the sets $W$ and $\calC$ of variables and constraints.
The resulting substructure of
the network for each constraint type is depicted in Figure \ref{netfig}.
\begin{itemize}
\item For each variable symbol $x \in W$, there is a unique input neuron $i_x$ and
an edge
to an immediate successor $j_x$.
The intended meaning
of these two neurons is to encode the value of the variable $x$ into the edge weight of $(i_x, j_x)$.
The number of neurons introduced in this step is $2|W|$.


\item For each constraint $C \in \calC$, an output neuron $o_C$ is inserted.
In addition, there are unique hidden
neurons $h_{(C, k)}$ and $q_{(C, k)}$, and a unique input neuron $p_{(C, k)}$ for each
$k \in \{1, \dots, l_C\}$, and
they are incident to the edges of the network as follows:
\begin{equation*}
(j_{x_k}, h_{(C, k)})\text{,} \quad (h_{(C, k)}, o_C)\text{,} \quad
(p_{(C, k)}, q_{(C, k)}) \quad \text{and} \quad (q_{(C, k)}, h_{(C, k)})\text{.}
\end{equation*}
In total, there are at most $10|\calC|$ new neurons introduced in this case.
Here, the neurons $p_{(C, k)}$ and $q_{(C, k)}$ are intended to cancel out the
output value of the neuron $h_{(C, k)}$ for those data points that in the construction
are not intended to directly concern the constraint $C$
but that still give a non-zero value for the input neuron $i_{x_k}$.

\item For each constraint of the form $x \times y + 1 = 0$, an additional
input neuron $e_C$ is added and connected with an edge $(e_C, j_y)$ to the neuron $j_y$.
For this, at most $|\calC|$ new neurons are introduced.

\end{itemize}
As the remaining part of the network architecture $N$,
the function $\varphi$ is selected so that for each neuron $v$,
the activation function $\varphi_v$ is the identity function except for the case
that $v$ is of type $h_{(C, 2)}$ for some
function constraint $C$ of the form $x_1 + f(x_2) = 0$;
in this case, the selection $\varphi_v \dfn f$ is used.
The set $A_E$ of active edges is defined to be the set of all edges of $N$ that are either
of the form
\begin{equation*}
(i_x, j_x)\text{,} \quad (j_y, h_{(C, 2)}) \quad \text{or} \quad (p_{(C', k)}, q_{(C', k)})\text{,}
\end{equation*}
where $x$ and $y$ are variables that appear in $W$,
$C$ is an inversion constraint $z \times y + 1 = 0$
for some variable $z \in W$, and $C' \in \calC$ is a constraint
and $k \in \{1, \dots, l_{C'}\}$ a corresponding index
in the sense of \ref{cindex}.
In Figure \ref{netfig}, active edges are drawn using thicker arrows
than the other edges of the network.

\newcommand{\onefig}{
\begin{frame}
\centering
    \begin{forest}
        for tree={
        circle, draw, inner sep=4pt,
        every label/.append style = {inner sep=0pt, font=\small},
        grow=west,
        s sep = 8mm,    
        l sep = 7mm,    
        edge = {latex-},
        }
    [,label=below:$o_C$
        [,label=below:$h_{(C, 1)}$
            [,label=below:$j_x$
                [,label=below:$i_x$,edge=very thick]
            ]
            [,label=below:$q_{(C, 1)}$
                [,label=below:$p_{(C, 1)}$,edge=very thick]
            ]
        ]
    ]
    \end{forest}
\end{frame}
}

\newcommand{\addfig}{
\begin{frame}
\centering
    \begin{forest}
        for tree={
        circle, draw, inner sep=4pt,
        every label/.append style = {inner sep=0pt, font=\small},
        grow=west,
        s sep = 8mm,    
        l sep = 7mm,    
        edge = {latex-},
        }
    [,label=280:$o_C$
        [,label=above:$h_{(C, 1)}$
            [,label=below:$j_x$
                [,label=below:$i_x$,edge=very thick]
            ]
            [,label=below:$q_{(C, 1)}$
                [,label=below:$p_{(C, 1)}$,edge=very thick]
            ]
        ]
        [,label=below:$h_{(C, 2)}$
            [,label=below:$j_y$
                [,label=below:$i_y$,edge=very thick]
            ]
            [,label=below:$q_{(C, 2)}$
                [,label=below:$p_{(C, 2)}$,edge=very thick]
            ]
        ]
        [,label=below:$h_{(C, 3)}$
            [,label=below:$j_z$
                [,label=below:$i_z$,edge=very thick]
            ]
            [,label=below:$q_{(C, 3)}$
                [,label=below:$p_{(C, 3)}$,edge=very thick]
            ]
        ]
    ]
    \end{forest}
\end{frame}
}

\newcommand{\mulfig}{
\begin{frame}
\centering
    \begin{forest}
        for tree={
        circle, draw, inner sep=4pt,
        every label/.append style = {inner sep=0pt, font=\small},
        grow=west,
        s sep = 8mm,    
        l sep = 7mm,    
        edge = {latex-},
        }
    [,label=280:$o_C$
        [,label=above:$h_{(C, 1)}$
            [,label=below:$j_x$
                [,label=below:$i_x$,edge=very thick]
            ]
            [,label=below:$q_{(C, 1)}$
                [,label=below:$p_{(C, 1)}$,edge=very thick]
            ]
        ]
        [,label=280:$h_{(C, 2)}$
            [,label=below:$j_y$,edge=very thick
                [,label=below:$i_y$,edge=very thick]
                [,label=below:$e_C$]
            ]
            [,label=below:$q_{(C, 2)}$
                [,label=below:$p_{(C, 2)}$,edge=very thick]
            ]
        ]
    ]
    \end{forest}
\end{frame}
}

\newcommand{\funfig}{
\begin{frame}
\centering
    \begin{forest}
        for tree={
        circle, draw, inner sep=4pt,
        every label/.append style = {inner sep=0pt, font=\small},
        grow=west,
        s sep = 8mm,    
        l sep = 7mm,    
        edge = {latex-},
        }
    [,label=below:$o_C$
        [,label=below:$h_{(C, 1)}$
            [,label=below:$j_x$
                [,label=below:$i_x$,edge=very thick]
            ]
            [,label=below:$q_{(C, 1)}$
                [,label=below:$p_{(C, 1)}$,edge=very thick]
            ]
        ]
        [,label=below:$h_{(C, 2)}$,label=above:$f$
            [,label=below:$j_y$
                [,label=below:$i_y$,edge=very thick]
            ]
            [,label=below:$q_{(C, 2)}$
                [,label=below:$p_{(C, 2)}$,edge=very thick]
            ]
        ]
    ]
    \end{forest}
\end{frame}
}

\begin{figure}
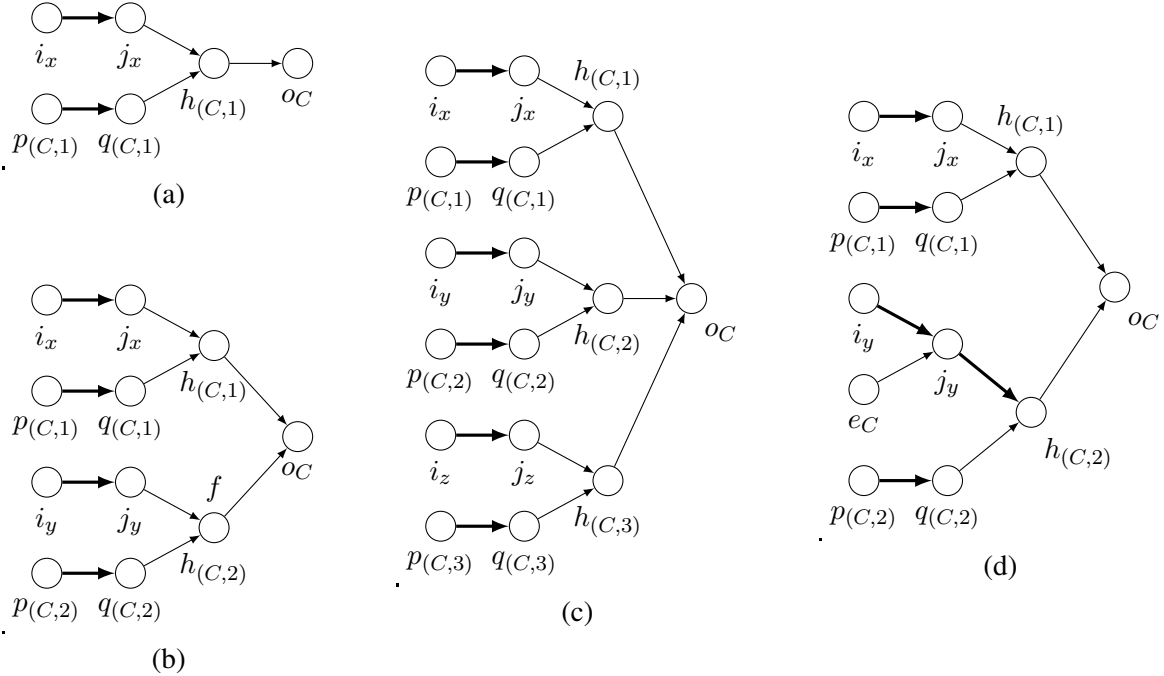

\begin{tabular}{l c r}
\begin{tabular}{c}
\onefig\\
(a)\\
\\
\\
\funfig\\
(b)
\end{tabular}
&
\begin{tabular}{c}
\addfig\\
(c)
\end{tabular}
&
\begin{tabular}{c}
\mulfig\\
(d)
\end{tabular}
\end{tabular}
\caption{(a) For a unit constraint $x = 1$, and for every other part of the network,
the value assigned to some variable $x$
is encoded in the weight of the edge from $i_x$ to $j_x$.
\quad (b) In the part of the network corresponding to a function constraint $x + f(y) = 0$,
the neuron $h_{(C, 2)}$ is activated using the given function $f$. For every other type
of neurons in the network, the identity activation function $x \mapsto x$ is used.
\quad (c) Input neurons $i_x$, $i_y$ and $i_z$, hidden neurons $j_x$, $j_y$, $j_z$ and
$h_{(C, k)}$ for $k \in \{1, 2, 3\}$, together with the output neuron $o_C$ enforce
an addition constraint $C$ of type $x + y + z = 0$. For each $k \in \{1, 2, 3\}$,
the neurons $p_{(C, k)}$ and $q_{(C, k)}$ are used to connect the corresponding variable
in the constraint with inversion constraints.
\quad (d) For an inversion constraint $C$ of the form $x \times y + 1 = 0$, an additional
input neuron, $e_C$, is introduced. An edge between the hidden levels of the network,
from $j_y$ to $h_{(C, 2)}$, is active in the training process. It is intended to carry
the negated value of the variable $x$.}
\label{netfig}
\end{figure}

Recall that all input neurons are either of the form $i_x$, $p_{(C, k)}$ or $e_C$, and
that $\{o_C \mid C \in \calC\}$ is the set of all output neurons.
Now, let $D$ consist of the following data points:
\begin{itemize}

\item There is a data point $d_a$ in $D$ that is defined as follows:
\begin{itemize}
\item For each variable symbol $x \in W$ it holds that $d_a(i_x) = 1$.
\item For each constraint $C$ of the form $x \times y + 1 = 0$ it
holds that $d_a(p_{(C, 1)}) = 1$ and $d_a(p_{(C, 2)}) = 1$.
\item For every other input neuron $v$, $d_a(v) = 0$.
\item For every output neuron $o_C$ corresponding to some unit
constraint $C \in \calC$, i.e., $C$
is $x = 1$ for some $x \in W$, it holds that $d_a(o_C) = 1$.
\item For every other output neuron $v$, we have $d_a(v) = 0$.
\end{itemize}

\item For each constraint $C \in \calC$ of the form $x \times y + 1 = 0$,
two additional data points, $d_{(C, 1)}$ and $d_{(C, 2)}$, are defined in the following steps:
\begin{itemize}
\item The data point $d_{(C, 1)}$ maps the tuple
$(i_x, i_y, e_C, o_C)$ of neurons to the tuple $(1, 0, 1, 0)$.
\item In addition, $d_{(C, 1)}$ maps to $1$ all those input
neurons of the form $p_{(C', k)}$ for which the number $k \in \{1, \dots, l_{C'}\}$
is an index for $x$ in some constraint $C' \neq C$.
\item In $d_{(C, 2)}$, the neurons of the tuple
$(i_x, i_y, e_C, o_C)$ are mapped to the values of ($0, 1, 0, 1)$.
\item The data point $d_{(C, 2)}$ has value $1$ for all input
neurons of the form $p_{(C', k)}$ for which $k \in \{1, 2, 3\}$
is an index for the variable $y$ in some $C' \in \calC$ such that $C' \neq C$.
\item If $v$ is any other input neuron or output neuron of $N$,
it holds that $d_{(C, 1)}(v) = d_{(C, 2)}(v) = 0$.
\end{itemize}
\end{itemize}
These steps introduce at most $2|\calC|$ data points.
Note that in the preceding definition of the data points $d_{(C, 1)}$ and $d_{(C, 2)}$, we use the assumption
that $\calC$ does not have constraints of the form $x \times x + 1 = 0$.

It remains to show that the $\NNtauTRAINING$ instance
$(N, A_E, A_V, D, c, \prec, \delta)$ has a satisfying solution
if and only if
the system $\calC$ is satisfiable.
For this, first assume that some edge weight function $w$ is a solution
for the training problem.
Let variable assignment $s$ be such that for each variable $x \in W$ it holds
that $s(x) = w(i_x, j_x)$. We show that $s$ satisfies $\calC$.
For this, let $C \in \calC$ be arbitrary.
\begin{itemize}
\item If $C$ is a unit constraint $x = 1$, then by the definition of $d_a$ and the non-negativeness of $c$,
it holds that $d_a(o_C) = 1$. Since $d_a(i_x) = 1$ and $d_a(p_{(C, 1)}) = 0$, we
have $w(i_x, j_x) = 1$ and so $s(x) = 1$.
\item In the case that $C$ is $x + y + z = 0$, we have $d_a(o_C) = 0$ and $d_a(p_{(C, k)}) = 0$
for each $k \in \{1, 2, ,3\}$. Thus, $w(i_x, j_x) + w(i_y, j_y) + w(i_z, j_z) = 0$ or,
equivalently, $s(x) + s(y) + s(z) = 0$.
\item If $C$ is some inversion constraint $x \times y + 1 = 0$,
the definition of the data point $d_{(C, 1)}$
leads to the following observation:
\begin{align*}
0 &= d_{(C, 1)}(o_C) \\
&= w(i_x, j_x) \times d_{(C, 1)}(i_x)
+ w(j_y, h_{(C, 2)}) \times (w(i_y, j_y) \times d_{(C, 1)}(i_y) + d_{(C, 1)}(e_C)) \\
&= w(i_x, j_x) \times 1 + w(j_y, h_{(C, 2)}) \times (w(i_y, j_y) \times 0 + 1) \\
&= w(i_x, j_x) + w(j_y, h_{(C, 2)})\text{.}
\end{align*}
Similarly, from the definition of $d_{(C, 2)}$ and the non-negativeness of the cost function $c$,
we get
\begin{align*}
1 &= d_{(C, 2)}(o_C) \\
&= w(i_x, j_x) \times d_{(C, 2)}(i_x)
+ w(j_y, h_{(C, 2)}) \times (w(i_y, j_y) \times d_{(C, 2)}(i_y) + d_{(C, 2)}(e_C)) \\
&= w(i_x, j_x) \times 0 + w(j_y, h_{(C, 2)}) \times (w(i_y, j_y) \times 1 + 0) \\
&= w(j_y, h_{(C, 2)}) \times w(i_y, j_y)\text{.}
\end{align*}
Thus it holds that $w(j_y, h_{(C, 2)}) = -w(i_x, j_x)$, and both of the claims
$w(i_x, j_x) \times w(i_y, j_y) + 1 = 0$
and $s(x) \times s(y) + 1 = 0$ are true.
\item If $C$ is $x + f(y) = 0$, we get $d_a(o_C) = 1$, $d_a(i_x) = d_a(i_y) = 1$ and
$d_a(p_{(C, k)}) = 0$ for $k \in \{1, 2\}$. From these it follows that
$w(i_x, j_x) + f(w(i_y, j_y)) = 0$ and $s(x) + f(s(y)) = 0$.
\end{itemize}
Therefore, the preceding assignment $s$ satisfies all the constraints that belong to $\calC$.

For the other direction, let $s$ be some assignment over the variables in $W$ such that
all constraints of $\calC$ are satisfied.
Let $w$ be an edge weight function defined on the set $A_E$ as follows:
\begin{itemize}
\item For each variable symbol $x \in W$, it holds that
$w(i_x, j_x) = s(x)$.
\item For each constraint $C$ of the form $x \times y + 1 = 0$,
weight value $-s(x)$ is assigned to the edge $(j_y, h_{(C, 2)})$.
The edge $(p_{(C, 2)}, q_{(C, 2)})$ is mapped to the value $s(x) \times s(y)$.
\item Every remaining $(p_{(C, k)}, q_{(C, k)}) \in A_E$
is given the weight value $-s(x_k)$, where $x_k$ is such that it matches with the index
$k \in \{1, 2, 3\}$ of constraint $C$ in the sense of \ref{cindex}.
\end{itemize}
Let $g$ be the neural function of the neural network $(N, w, \varnothing)$. It remains
to show that for every output neuron $v$ and for every data point $d \in D$ it holds that
$g_d(v) = d(v)$. Recall that for each output neuron $v$ there is a unique
$C \in \calC$ such that $v = o_C$.

We will first check the case of the data point $d_a$.
If $C \in \calC$ is not an inversion constraint, then for all indices
$k \in \{1, \dots, l_C\}$ as in \ref{cindex} it holds that
\begin{equation*}
g_{d_a}(h_{(C, k)})
= w(i_{x_k}, j_{x_k}) \times d_a(i_{x_k}) + w(p_{(C, k)}, q_{(C, k)}) \times d_a(p_{(C, k)})
= s(x_k) \times 1 + 0
= s(x_k)\text{.}
\end{equation*}
From this it follows that if $C$ is a unit constraint, then $g_{d_a}(o_C) = s(x_1) = 1$,
and if $C$ is an addition constraint, then
$g_{d_a}(o_C) = s(x_1) + s(x_2) + s(x_3) = 0$, and furthermore, if $C$ is a function
constraint of the form $x_1 + f(x_2) = 0$, then $g_{d_a}(o_C) = s(x_1) + f(x_2) = 0$.
That is, for each of the preceding cases, it holds that $g_{d_a}(o_C) = d_a(o_C)$.
On the other hand, if $C$ is an inversion constraint $x \times y + 1 = 0$, then
\begin{align*}
g_{d_a}(h_{(C, 1)})
&= w(i_x, j_x) \times d_a(i_x) + w(p_{(C, 1)}, q_{(C, 1)}) \times d_a(p_{(C, 1)}) \\
&= s(x) \times 1 + (-s(x)) \times 1
= 0\text{,}
\end{align*}
and
\begin{align*}
g_{d_a}(h_{(C, 2)})
&= \left(w(j_y, h_{(C, 2)}) \times (w(i_y, j_y) \times d_a(i_y) + d_a(e_C))\right)
+ w(p_{(C, 2)}, q_{(C, 2)}) \times d_a(p_{(C, 2)}) \\
&= \left(-s(x) \times (s(y) \times 1 + 0) + (s(x) \times s(y)) \times 1\right)
= 0\text{.}
\end{align*}
Thus, it holds that $g_{d_a}(o_C) = 0 = d_a(o_C)$.

For the remaining cases, let $C$ be a fixed inversion constraint $x \times y + 1 = 0$.
We will go through the data points $d_{(C, 1)}$ and $d_{(C, 2)}$ simultaneously.
First, let $d \in \{d_{(C, 1)}, d_{(C, 2)}\}$ and
$C' \in \calC \setminus \{C\}$ be arbitrary.
Then, for each $k \in \{1, \dots, l_{C'}\}$ it holds that $h_{(C', k)} = 0$,
namely, either it holds that both $d(x_k) = 0$ and $d(p_{(C', k)}) = 0$ are true, or
that both $d(x_k) = 1$ and $d(p_{(C', k)}) = 1$ are true.
Therefore, we obtain the result $g(o_{C'}) = 0 = d(o_{C'})$.

For the inversion constraint $C$ and neuron $h_{(C, 1)}$, for each
$d \in \{d_{(C, 1)}, d_{(C, 2)}\}$ it holds that
\begin{equation*}
g_d(h_{(C, 1)})
= w(i_x, j_x) \times d(i_x)
+ w(p_{(C, 1)}, q_{(C, 1)}) \times d(p_{(C, 1)})
= s(x) \times d(i_x) + 0\text{.}
\end{equation*}
From this we get $g_{d_{(C, 1)}}(h_{(C, 1)}) = s(x)$ and $g_{d_{(C, 2)}}(h_{(C, 1)}) = 0$.
On the other hand, for the hidden neuron $h_{(C, 2)}$
and $d \in \{d_{(C, 1)}, d_{(C, 2)}\}$ we have
\begin{align*}
g_d(h_{(C, 2)})
&= \left(w(j_y, h_{(C, 2)}) \times (w(i_y, j_y) \times d(i_y) + d(e_C))\right)
+ w(p_{(C, 2)}, q_{(C, 2)}) \times d(p_{(C, 2)}) \\
&= -s(x) \times (s(y) \times d(i_y) + d(e_C)) + 0 \\
&= -s(x) \times s(y) \times d(i_y) - s(x) \times d(e_C)\text{.}
\end{align*}
Then, $g_{d_{(C, 1)}}(o_C) = s(x) - s(x) = 0$ and
$g_{d_{(C, 2)}}(o_C) = 0 - s(x) \times s(y) = 1$, from which it follows that
$g_d(o_C) = d(o_C)$ for each $d \in \{d_{(C, 1)}, d_{(C, 2)}\}$.
This concludes the proof the theorem.
\end{proof}


In the previous proof, it is enough to have a single unit constraint
in the set $\calC$. Addition constraints can be combined under
a unified output neuron, if multiple data points are allowed in order to take
separate addition constraints into account. Here, we can use the fact that
the indices of the input neurons $p_{(C, k)}$ are based on the numbers $k$
instead of the corresponding variable $x_k$; multiple occurrences of a variable in
an unified addition constraint can be treated separately.
Furthermore, as noted in \cite{AbrahamsenKM21}, it can be
assumed that each variable is incident to at most one inversion constraint.
In this manner, only a constant number of data points is needed for all the
inversion constraints that are in $\calC$.


As a reminder, the sigmoid activation function $\sigma \colon \RE \to \RE$
is defined as $\sigma(x) = 1 / (1 + \exp(-x))$ for
each $x \in \RE$. Using the previous theorem, we can now connect sigmoidal NN training
with exponential real arithmetic.

\begin{corollary}
$\NNsigmoidTRAINING$ is $\exexp$-complete.
\end{corollary}
\begin{proof}
By Theorem \ref{NNcomplete}, $\NNsigmoidTRAINING$ is $\ex_\sigma$-complete. Furthermore,
the sigmoid activation function is existentially definable in $\exexp$
and similarly the exponential function is existentially definable in $\ex_\sigma$.
Thus it follows that $\NNsigmoidTRAINING$ is $\exexp$-complete.
\end{proof}

\begin{corollary}
$\NNsinTRAINING$ is undecidable but in $\Sigma^0_3$.
\end{corollary}

\begin{proof}
The claim follows from Theorems \ref{sincomplete} and \ref{NNcomplete}.
\end{proof}

Similar to the proof of Theorem \ref{sincomplete},
undecidability can be generalized to other functions having
some form of periodic nature that can be used to define natural numbers with
existential formulae.

\section{Conclusion}
We have studied the complexity of the neural network training problem.
We showed that the training problem corresponding to sets $\tau$ of activation functions are
complete for the complexity classes $\ex_\tau$ that extend the reals with the corresponding functions. We also showed that for sets $\tau$ of
effectively continuous functions, these problems reside in $\Sigma^0_3$,
and even in $\Sigma^0_1$ when defined without the equality sign and negation. The following topics deserve further study:
\begin{itemize}
\item 
Could our $\Sigma^0_1$ upper bound for the identity-free existential theory of the reals with effectively continuous functions be improved or extended.
Note that for decidability of the exponential arithmetic it would suffice to establish an $\Sigma^0_1$-upper bound for its existential fragment \cite{Wilkie1997,Servi08}.
\item The proof of our main completeness result (Theorem \ref{NNcomplete}) is based on the use of adversarial
network architectures (see \cite{Bertschinger22}). It is an open question whether
our result holds if the NN architecture is assumed to be a fully connected graph
with all edge weights and neuron biases being active, or when restricted to the case that all neurons use the same activation function.


\item For many non-negative cost functions it is known that
the computational complexity of an exact total training error $0$ is at
most the complexity of an error bounded by some positive value (allowing equality),
whereas
an upper bound given by the strict order relation $<$ seems to be independent of the two.
Recall that these cases correspond to our upper bounds $\Sigma^0_3$ and $\Sigma^0_1$,
and that in the case of the exponential function, they are also connected to
Tarski's exponential function problem.

\item Another topic to study are training problems with limited precision.
For example, does the complexity of the training problem decrease for some activation functions if edge weights and biases are assumed to be rational?

\end{itemize}






%

\bibliographystyle{IEEEtran}
\bibliography{biblio}

\end{document}